\documentclass{amsart}

\usepackage{graphicx,tikz,amssymb,mathtools,tabularx,algorithm,enumerate,enumitem}
\usetikzlibrary{arrows,calc,patterns}
\usepackage[justification=centering]{caption}
\usepackage[noend]{algpseudocode}

\makeatletter
\def\BState{\State\hskip-\ALG@thistlm}
\makeatother

\DeclareMathOperator{\EL}{EL}

\DeclareMathOperator{\B}{M}
\DeclareMathOperator{\ff}{First-Fit}

\newtheorem{theorem}{Theorem}[section]
\newtheorem{defin}[theorem]{Definition}
\newtheorem{lemma}[theorem]{Lemma}

\definecolor{gg}{rgb}{0.0,0.5,0.0}
\definecolor{amber}{rgb}{1.0,0.75,0.0}
\definecolor{amy}{rgb}{0.6,0.4,0.8}

\author{Israel R. Curbelo}
\email{icurbelo@kean.edu}

\address{Department of Mathematical Sciences, Kean University, Union, NJ 07083}

\title{Online coloring of short interval graphs and two-count interval graphs}

\begin{document}

\begin{abstract}
We study the online coloring of $\sigma$-interval graphs, which are interval graphs with interval lengths in $[1,\sigma]$ and 2-count interval graphs, which are interval graphs that require at most two distinct interval lengths.

For $\sigma$-interval graphs, the Kierstead-Trotter algorithm has competitive ratio 3 and no online algorithm has competitive ratio better than 2. In this paper, we show that for every $\varepsilon>0$, there is a $\sigma>1$ such that there is no online algorithm for $\sigma$-interval coloring with competitive ratio less than $3-\varepsilon$. 

For 2-count interval graphs, we show that the greedy algorithm First-Fit has competitive ratio at most $4$, that there is no online algorithm with competitive ratio less than $2.5$ when the interval representation is unknown, and that there is no online algorithm with competitive ratio less than $2$ when the interval representation is known.     
\end{abstract}

\maketitle

\section{Introduction}

An online graph coloring algorithm $A$ receives a graph $G$ in the order of its vertices $v_1,\ldots,v_n$ and constructs an online coloring. This means that the color assigned to a vertex $v_i$ depends solely on the subgraph induced by the vertices $\{v_1,\ldots,v_i\}$ and the colors assigned to them. One of the simplest online coloring algorithms is the greedy algorithm $\ff$, which traverses the list of vertices given in the order they are presented, and assigns each vertex the smallest color not assigned to its neighbors that appear before it in the list of vertices. 

The performance of an online coloring algorithm $A$ on a graph $G$ is measured with respect to the chromatic number $\chi(G)$, that is, the number of colors used by an optimal offline algorithm. Let $\chi_A(G)$ denote the maximum number of colors needed by algorithm $A$ over all orderings of the vertices of $G$. An online coloring algorithm $A$ is \emph{$\rho$-competitive} if there exists a $\beta\in\mathbb{R}$ such that for every graph $G$, $\chi_A(G)\leq\rho\cdot\chi(G)+\beta$. The \emph{competitive ratio} of algorithm $A$ is $\inf\{\,\rho\mid\text{$A$ is $\rho$-competitive}\,\}$. In general, $\ff$ does not have a bounded competitive ratio, and there is no online algorithm with a constant competitive ratio. The best known algorithm \cite{hal-97,hal-sze-94} uses $O(\chi n/\log n)$ colors for $\chi$-colorable graphs of size $n$, and no online coloring algorithm can be $o(n/\log^2 n)$-competitive. For this reason, it is common to restrict the class of graphs. An important class of intersection graphs that has been extensively studied in this context, with applications in scheduling and resource allocation, is the class of interval graphs. 

A graph $G=(V,E)$ is an \emph{interval graph} if there is a function $I$ which assigns to each vertex $v\in V$ a closed interval $I(v)$ on the real line so that for all $u,v\in V$, $u$ and $v$ are adjacent if and only if $I(u)\cap I(v)\neq\emptyset$. We refer to $I$ as an \emph{interval representation} of $G$.  When every interval has unit length, we refer to $G$ as a \emph{unit-interval graph} and $I$ as a \emph{unit-interval representation} of $G$. For $\sigma\geq 1$, we define a $\sigma$-interval representation as an interval representation where every interval has length in $[1,\sigma]$. We say a graph is a $\sigma$-interval graph if it has a $\sigma$-interval representation. When $\sigma=1$, this is the class of unit-interval graphs, and when ``$\sigma=\infty$'', this is the class of interval graphs. The interval count of an interval graph $G$ is the smallest integer $k\geq1$ such that $G$ has an interval representation with intervals of at most $k$ distinct lengths. When $k=1$, this is the class of unit-interval graphs. In other words, we study intermediate cases of unit-interval graphs and interval graphs. Furthermore, for any $1<\sigma<\infty$ and for any $1<k<\infty$, the class of $\sigma$-interval graphs and the class of $k$-count interval graphs are not the same. It is easy to see that the star $K_{1,\lceil{\sigma\rceil}+2}$ is a 2-count interval graph but not a $\sigma$-interval graph. For an example of a graph that is a $(1+\varepsilon)$-interval graph for any $\varepsilon>0$ but not a 2-count interval graph, see Figure \ref{fig:2count}. Lastly, if $G$ is an interval graph, then $\chi(G)=\omega(G)$ where $\omega(G)$ denotes the clique number of the graph $G$.

\begin{figure}
    \begin{center}
    \begin{tikzpicture}[scale=0.7]
    \node (e1) at (-1,-0.5){1};
    \node (e) at (0,-0.5){2};
    \node (f) at (1,-0.5){3};
    \node (g) at (3,-0.5){$\lceil\sigma\rceil$};
    \node (g1) at (5,-0.5){$\lceil\sigma\rceil+1$};
    \node (g2) at (7,-0.5){$\lceil\sigma\rceil+2$};
    \node (d) at (2,0){$\ldots$};
    \node (K) at (2,-2){$K_{1,\lceil{\sigma\rceil}+2}$};
    \node (J) at (12,-2){$H$};
    \begin{scope}[every node/.style={circle,draw,fill=black,inner sep=2pt}]

    \node (E0) at (-1,0){};
    \node (E) at (0,0){};
    \node (F) at (1,0){};
    \node (G) at (3,0){};
    \node (G1) at (5,0){};
    \node (G2) at (7,0){};
    \node (H) at (2,1){};

    \node (A) at (12,1){};
    \node (B) at (10,0){};
    \node (B1) at (11,0){};
    \node (B2) at (12,-1){};
    \node (B3) at (13,0){};
    \node (B4) at (14,0){};
    \node (B5) at (15,0){};
    
    \end{scope}

    \begin{scope}[every edge/.style={draw, thick}]

    \path (H) edge (E0);
    \path [-] (H) edge (E);
    \path [-] (H) edge (F);
    \path [-] (H) edge (G);
    \path [-] (H) edge (G1);
    \path [-] (H) edge (G2);

    \path (A) edge (B); 
    \path (A) edge (B1);
    \path (A) edge (B2);
    \path (A) edge (B3);
    \path (A) edge (B4);
    \path (A) edge (B5);
    \path (B) edge (B1);
    \path (B1) edge (B2);
    \path (B2) edge (B3);
    \path (B2) edge (B4);
    \path (B4) edge (B5);
            
    \end{scope}
    \end{tikzpicture}
    \end{center}
    \caption{The star $K_{1,\lceil{\sigma\rceil}+2}$ is a 2-count interval graph but not a $\sigma$-interval graph. The graph $H$ is a $(1+\varepsilon)$-interval graph for any $\varepsilon>0$ but not a 2-count interval graph.}
    \label{fig:2count}
\end{figure}

\subsection{Previous Work}

In 1976, Woodall \cite{woo-74} studied the performance of $\ff$ on interval graphs. After extensive research on this problem \cite{kie-88,kie-qin-95,kie-tro-81,chr-slu-88,pem-ram-04,nar-sub-08,kie-smi-16}, the competitive ratio of $\ff$ has been shown to be between 5 and 8. In 1981, Kierstead and Trotter \cite{kie-tro-81} presented an algorithm with a competitive ratio of 3. Additionally, they presented a strategy for constructing an interval graph that forces any online coloring algorithm to use $3\omega-2$ colors on an interval graph of clique number at most $\omega$. 
\begin{theorem}[Kierstead and Trotter \cite{kie-tro-81}]
    There is an online algorithm that uses at most $3\omega-2$ colors to color any interval graph with clique number at most $\omega$. Additionally, there is no online algorithm that uses less than $3\omega-2$ colors on any interval graph with clique number $\omega$.
    \label{thm:kt}
\end{theorem}

In 1988, Chrobak and \'Slusarek \cite{chr-slu-88} considered the same problem except that the online algorithm received the interval representation as input. Despite the additional information provided to the algorithm, they arrived at the same conclusion. They also considered the problem restricted to unit-interval graphs. They showed that the algorithm $\ff$ never uses more than $2\omega -1$ colors on a unit-interval graph of clique number at most $\omega$. Additionally, they presented a strategy that forces $\ff$ to use $2\omega-1$ colors on a unit-interval graph of clique number $\omega$.
\begin{theorem}[Chrobak and \'{S}lusarek \cite{chr-slu-88}]
    The algorithm $\ff$ uses at most $2\omega-1$ colors to color any unit-interval graph of clique number at most $\omega$, and there is a unit-interval graph of clique number $\omega$ for which $\ff$ uses $2\omega-1$ colors.
    \label{thm:cs}
\end{theorem}
In 2005, Epstein and Levy \cite{eps-lev-05} constructed a strategy presenting unit intervals that forces any online algorithm to use $3k$ colors on a unit-interval graph of clique number at most $2k$. 
\begin{theorem}[Epstein and Levy \cite{eps-lev-05}]
    There is no online algorithm that uses less than $\lfloor1.5\omega\rfloor$ colors on any unit-interval graph of clique number at most $\omega$ with known interval representation.
    \label{thm:el}
\end{theorem}
In 2023, Bir\'o and Curbelo \cite{bir-cur-23} showed that the strategy could be improved to force one more color for odd clique numbers. 

In 2008, Micek presented a strategy that forces any online coloring algorithm to use at least $2\omega-1$ colors on a unit-interval graph of clique number at most $\omega$. A proof was later published in the survey paper by Bosek et al. \cite{survey}.

\begin{theorem}[Micek; Bosek et al. \cite{survey}]
    There is no online algorithm that uses less than $2\omega-1$ colors on any unit-interval graph of clique number at most $\omega$.
\end{theorem}

In 2024, Chybowska-Sok\'ol et al. \cite{cgjmp-24} studied 
online coloring of $\sigma$-interval graphs and showed that the competitive ratio can be arbitrarily close to $2.5$ for increasing values of $\sigma$.

\begin{theorem}[Chybowska-Sok\'ol et al. \cite{cgjmp-24}]
    For every $\varepsilon>0$, there is a $\sigma>1$ such that there is no online algorithm with competitive ratio less than $2.5-\varepsilon$ for online coloring of $\sigma$-interval graphs with known representation.
\end{theorem}

\subsection{Our work}

In this paper, we show that for every $\varepsilon>0$, there is a $\sigma>1$ such that there is no online algorithm with competitive ratio less than $3-\varepsilon$ for $\sigma$-interval coloring. 
\begin{theorem}
    For every $\varepsilon>0$, there is a $\sigma>1$ such that there is no online algorithm with competitive ratio less than $3-\varepsilon$ for coloring of $\sigma$-interval graphs.
    \label{thm:1}
\end{theorem}

Additionally, we study the online coloring of 2-count interval graphs. In this setting, we consider both online algorithms that receive the interval representation as input and online algorithms that do not. We show that there is no online algorithm with competitive ratio less than $2.5$ when the interval representation is unknown, and that there is no online algorithm with competitive ratio less than $2$ when the interval representation is known.   

\begin{theorem}
    For every $\varepsilon>0$, there is no $(2.5-\varepsilon)$-competitive online coloring algorithm for 2-count interval graphs.
    \label{thm:two-len}
\end{theorem}

\begin{theorem}
    For every $\varepsilon>0$, there is no $(2-\varepsilon)$-competitive online coloring algorithm for 2-count interval graphs with known interval representation. 
    \label{thm:final}
\end{theorem}


In Section \ref{sec:sig}, we prove our result for $\sigma$-interval graphs, and in Section \ref{sec:two}, we prove our results for 2-count interval graphs.

\section{Online coloring of $\sigma$-interval graphs}\label{sec:sig}

\begin{figure}
    \centering
    \begin{tikzpicture}[yscale=3,xscale=0.74,thick, fill opacity=0.2]
        \foreach \x in {1,...,14}{
            \draw[fill opacity =1] (\x,1.03) -- (\x,.97) node[below]{\x};
        }
        \draw[fill opacity =1] (15.5,1.03) -- (15.5,0.97)node[below]{$\infty$};
        \foreach \y in {1,...,3}{
            \draw[fill opacity =1] (1.1,\y) -- (0.9,\y) node[left]{\y};
        }
        \draw[blue,fill=blue] (5,1.5) rectangle (5.7,1.6);
        \node[fill opacity=1] (b) at (5.7,1.55)[right]{unknown interval representation (NEW)};
        \draw[red,fill=red] (5,1.35) rectangle (5.7,1.45);
        \node[fill opacity=1] (b) at (5.7,1.4)[right]{known interval representation \cite{chr-slu-88,cgjmp-24}};
        \draw[dotted] (5,1.25) -- (5.7,1.25);
        \node[fill opacity=1] (b) at (5.7,1.25)[right]{previously best lower bounds on $\ff$ \cite{cgjmp-24}};
        
        \fill[red] (1,1.66) -- (2,1.66) -- (2,1.75) -- (8,1.75) -- (8,1.81) -- (14.3,1.81) -- (15,2.5) -- (16,2.5) -- (16,3) -- (14.3,3) -- (2,3) -- (1,2) -- (1,1.66);
        \fill[blue] (1,2) -- (2,2) -- (2,2.25) -- (3,2.25) -- (3,2.5) -- (8,2.5) -- (8,2.625) -- (12,2.625) -- (12,2.75) -- (14.3,2.75) -- (15,3) -- (14.3,3) -- (1,3) -- (1,2);
        \draw[draw = white,fill=white, fill opacity=1] (14.3,1.8) rectangle (15,3);
        \draw[blue, dotted, domain=14.3:15, samples=100]
  plot (\x, {3 + 0.25/(ln(10)-ln(3))*ln(\x-14)});
        \draw[red, dotted, domain=14.3:15, samples=100]
  plot (\x, {2.5 + 0.69/(ln(10)-ln(3))*ln(\x-14)});
        \draw[red, fill opacity=1] (15,2.5) -- (16,2.5)node[right]{\cite{cgjmp-24}};

        \draw[fill opacity=1] (7.5,3) node[above]{\textcolor{red}{\cite{chr-slu-88,cgjmp-24}}};
        \draw[fill opacity=1] (14.3,1.81) node[right]{\textcolor{red}{\cite{cgjmp-24}}};
        \draw[fill opacity=1] (3,2) node[right]{\textcolor{black}{\cite{cgjmp-24}}};
        \draw[fill opacity=1] (12,2.5) node[right]{\textcolor{black}{\cite{cgjmp-24}}};
        \draw[fill opacity=1] (1.5,3) node[above]{\textcolor{blue}{\cite{kie-tro-81}}};
        \draw[fill opacity=1] (1,1) -- (1,3.2) node[above]{ratio};
        \draw[fill opacity=1] (1,1) -- (14.3,1) node[right]{};
        \node at (14.65,1){$\ldots$};
        \draw[fill opacity=1] (15,1) -- (16,1) node[right]{$\sigma$};
        \draw[dotted] (1,2) -- (3,2);
        \node[fill=black, fill opacity=1, circle, transform shape=false, inner sep=1.5pt] at (3,2){};
        \draw[dotted] (3,2.5) -- (12,2.5);
        \node[fill=black, fill opacity=1, circle, transform shape=false, inner sep=1.5pt] at (12,2.5){};
        
        \draw[red] (1,2) -- (2,3) -- (9,3);
        \draw[red] (1,1.66) -- (2,1.66);
        \node[fill=red, fill opacity=1, circle, transform shape=false, inner sep=1.5pt] at (2,1.66){};
        \draw[red] (2,1.75) -- (8,1.75);
        \draw[red] (8,1.81) -- (14.3,1.81);
        \node[draw=red, fill=white, fill opacity=1, circle, transform shape=false, inner sep=1.5pt] at (2,1.75){};
        \node[fill=red, fill opacity=1, circle, transform shape=false, inner sep=1.5pt] at (8,1.75){};
        \node[draw=red, fill=white, fill opacity=1, circle, transform shape=false, inner sep=1.5pt] at (8,1.81){};
        \node[draw=red, fill=white, fill opacity=1, circle, transform shape=false, inner sep=1.5pt] at (1,1.66){};
        
        \draw[blue] (8,2.625) -- (12,2.625);
        \draw[blue] (1,3) -- (14.3,3);
        \draw[violet] (2,3) -- (16,3);
        \draw[blue] (12,2.75) -- (14.3,2.75);
        \draw[blue] (1,2) -- (2,2);
        \draw[blue] (2,2.25) -- (3,2.25);
        \draw[blue] (3,2.5) -- (8,2.5);

        \node[draw=blue, fill=white, fill opacity=1, circle, transform shape=false, inner sep=1.5pt] at (8,2.625){};
        \node[fill=blue, fill opacity=1, circle, transform shape=false, inner sep=1.5pt] at (12,2.625){};
        \node[fill=blue, fill opacity=1, circle, transform shape=false, inner sep=1.5pt] at (8,2.5){};
        \node[draw=blue, fill=white, fill opacity=1, circle, transform shape=false, inner sep=1.5pt] at (12,2.75){};
        \node[fill=blue, fill opacity=1, circle, transform shape=false, inner sep=1.5pt] at (1,2){};
        \node[fill=blue, fill opacity=1, circle, transform shape=false, inner sep=1.5pt] at (2,2){};
        \node[draw=blue, fill=white, fill opacity=1, circle, transform shape=false, inner sep=1.5pt] at (2,2.25){};
        \node[fill=blue, fill opacity=1, circle, transform shape=false, inner sep=1.5pt] at (3,2.25){};
        \node[draw=blue, fill=white, fill opacity=1, circle, transform shape=false, inner sep=1.5pt] at (3,2.5){};
        \node[fill=blue, fill opacity=1, circle, transform shape=false, inner sep=1.5pt] at (1,2){};
    \end{tikzpicture}
    \caption{Gap comparison between the competitive ratios for online coloring of $\sigma$-interval graphs with known interval representation and with unknown interval representation.}
    \label{fig:gap}
\end{figure}

In this section, we study the online coloring of $\sigma$-interval graphs. Theorem \ref{thm:kt} provides an upper bound of $3$ on the competitive ratio for online coloring of $\sigma$-interval graphs for every $\sigma\geq 1$. Chybowska-Sok\'ol et al. \cite{cgjmp-24} provided non-trivial lower bounds on the competitive ratio for algorithms with known interval representation. Since withholding the interval representation cannot help the online algorithm, every lower bound from the known-representation setting also applies when the representation is unknown. They showed that for any $\varepsilon>0$, there is a $\sigma>1$ such that the competitive ratio for online coloring of $\sigma$-interval graphs with known representation is at least $\frac{5}{2}-\varepsilon$. We improve this to $3-\varepsilon$ when the interval representation is unknown to the algorithm. See Figure \ref{fig:gap} for a comparison of the results for $\sigma\in[1,14]$ and for $\sigma=\infty$.

\begin{defin}
    Let $\omega$ be a positive integer, let $S_\omega$ be a strategy for constructing an interval graph $G$ of clique number $\omega$, and let $\rho$ and $\beta$ be real numbers. If \(S_\omega\) forces any online coloring algorithm to produce a component of \(G\) on which \(\rho\omega+\beta\) colors are used, and $G$ has an interval representation $I$ such that the image of each component of $G$ under $I$ is contained in an interval of length $\lambda$, then we call $S_\omega$ a \emph{$\langle\omega,\rho,\beta,\lambda\rangle$-strategy}.
\end{defin}

\begin{lemma}\label{lemma:strat}
    Let $\rho$, $\beta$ and $\lambda$ be constants, $\omega$ be a positive integer, $n$ and $N$ be non-negative integers, and let $S_\omega$ be an $\langle{\omega,\rho,\beta,\lambda\rangle}$-strategy. For any $\delta>0$, there is a strategy $S$ such that $S$ constructs a $(4^{n-1}\lambda+\delta/4)$-interval graph, and either
    \begin{enumerate}
        \item[(i)] \(S\) is a \(\langle 2^n\omega, \rho + 3(2^n - 1), \beta, 4^n \lambda + \delta \rangle\)-strategy, or
        \item[(ii)] \(S\) is a \(\langle 2^n\omega, N, \beta, \Lambda \rangle\)-strategy for some constant \(\Lambda\).
    \end{enumerate}
\end{lemma}
    
\begin{proof}

     \begin{figure}
        \centering
        \begin{tikzpicture}[yscale=0.6,xscale=0.44]

        \def\xA{-1}
        \def\xB{0}
        \def\xC{4}
        \def\xD{5}
        \def\xE{5}
        \def\xF{6}
        \def\xG{10}
        \def\xH{11}
        \def\xI{11}
        \def\xJ{12}
        \def\xK{16}
        \def\xL{17}
        \def\xM{17}
        \def\xN{18}
        \def\xO{22}
        \def\xP{23}
        
        \draw[dotted] (\xA,1.2)--(\xA,5.5) node[above]{$\scriptstyle 0$};
        \draw[dotted] (\xD,1.2)--(\xD,5.5) node[above]{$\scriptstyle 4^n\lambda+\frac{\delta}{4}$};
        \draw[dotted] (\xH,2.4)--(\xH,5.5) node[above]{$\scriptstyle 2(4^n\lambda+\frac{\delta}{4})$};
        \draw[dotted] (\xL,1.2)--(\xL,5.5) node[above]{$\scriptstyle 3(4^n\lambda+\frac{\delta}{4})$};
        \draw[dotted] (\xP,1.2)--(\xP,5.5) node[above]{$\scriptstyle 4(4^n\lambda+\frac{\delta}{4})$};
        
        \fill[yellow!50!] (\xB,0) rectangle (\xC,1.1);
        \node at ({(\xB+\xC)/2},0.5){$\scriptstyle H_a$};
        
        \fill[yellow!50!] (\xF,0) rectangle (\xG,1.1);
        \node at ({(\xF+\xG)/2},0.5){$\scriptstyle H_{1}$};
        
        \fill[yellow!50!] (\xJ,0) rectangle (\xK,1.1);
        \node at ({(\xJ+\xK)/2},0.5){$\scriptstyle H_{2}$};
        
        \fill[yellow!50!] (\xN,0) rectangle (\xO,1.1);
        \node at ({(\xN+\xO)/2},0.5){$\scriptstyle H_{b}$};
        
        \fill[red!40!] (\xA,1.2) rectangle (\xE,2.3);
        \path[thick,pattern=horizontal lines, pattern color = white] (\xA,1.2) rectangle (\xE,2.3);
        \node at ({(\xA+\xE)/2},1.7){$\scriptstyle K_{\omega_n}^{a}$};
        
        \fill[red!40!] (\xL,1.2) rectangle (\xP,2.3);
        \path[thick,pattern=horizontal lines, pattern color = white] (\xL,1.2) rectangle (\xP,2.3);
        \node at ({(\xL+\xP)/2},1.7){$\scriptstyle K_{\omega_n}^{b}$};
        
        \fill[blue!40!] (\xD,2.4) rectangle (\xI,3.5);
        \path[thick,pattern=horizontal lines, pattern color = white] (\xD,2.4) rectangle (\xI,3.5);
        \node at ({(\xD+\xI)/2},2.9){$\scriptstyle K_{\omega_n}^{1}$};
        
        \fill[green!40!] (\xH,3.6) rectangle (\xM,4.7);
        \path[thick,pattern=horizontal lines, pattern color = white] (\xH,3.6) rectangle (\xM,4.7);
        \node at ({(\xH+\xM)/2},4.1){$\scriptstyle K_{\omega_n}^{2}$};
        
        \draw[|-|] ({\xB-0.5},0) -- ({\xB-0.5},1.1);
        \node[left] at ({\xB-0.5},0.5){$\scriptstyle {2^n\omega}$};
        
        \draw[|-|] ({\xA-0.5},1.2) -- ({\xA-0.5},2.3);
        \node[left] at ({\xA-0.5},1.7){$\scriptstyle {2^n\omega}$};
        
        \draw[|-|] ({\xD-0.5},2.4) -- ({\xD-0.5},3.5);
        \node[left] at ({\xD-0.5},2.9){$\scriptstyle {2^n\omega}$};
        
        \draw[|-|] ({\xH-0.5},3.6) -- ({\xH-0.5},4.7);
        \node[left] at ({\xH-0.5},4.1){$\scriptstyle {2^n\omega}$};
        
        \foreach \L/\R in {\xB/\xC, \xF/\xG, \xJ/\xK, \xN/\xO}{
            \draw[|-|] (\L,-0.5) -- (\R,-0.5);
            \node[below] at ({(\L+\R)/2},-0.5){$\scriptstyle {4^n\lambda+\frac{\delta}{12}}$};
        }
        
        \draw[|-|] (\xA,-2) -- (\xP,-2);
        \node[below] at ({(\xA+\xP)/2},-2){$\scriptstyle 4^{n+1}\lambda+\delta$};
        
        \end{tikzpicture}
        \caption{An interval representation for the component $H$ constructed by the $\langle{2^n\omega,\rho+3(2^n-1),\beta,4^n\lambda+\delta\rangle}$-strategy in Lemma \ref{lemma:strat} with $[\rho +3(2^{n+1}-1)]\omega+\beta$ colors.}
        \label{fig:bridge}
    \end{figure}
        
    Let $\rho$, $\beta$ and $\lambda$ be constants, let $\omega$ be a positive integer, let $N$ be a non-negative integer, and let $S_\omega$ be an $\langle{\omega,\rho,\beta,\lambda\rangle}$-strategy. For each positive integer $k$, let $\omega_k=2^k\omega$, $\rho_k=\rho+3(2^k-1)$, $\lambda_k=4^k\lambda$ and $\delta'=\delta/12$. 
    If $n=0$, then the proof is trivial. Hence, we assume that for some integer $n\geq 1$, there is a strategy $S_{\omega_n}$ such that $S_{\omega_n}$ constructs a $(\lambda_{n-1}+3\delta')$-interval graph, and $S_{\omega_n}$ is a $\langle{\omega_n,\rho_n,\beta,\lambda_n+\delta'\rangle}$-strategy or $S_{\omega_n}$ is a $\langle{\omega_n,N,\beta,\Lambda\rangle}$-strategy for some $\Lambda\in\mathbb{R}$. We prove the result inductively on $n$. If $S_{\omega_n}$ is a $\langle{\omega_n,N,\beta,\Lambda\rangle}$-strategy for some $\Lambda\in\mathbb{R}$, then we are done. Therefore, we assume that $S_{\omega_n}$ is a $\langle{\omega_n,\rho_n,\beta,\lambda_n+\delta'\rangle}$-strategy.
    If $N\leq \rho_n$, then $S_{\omega_n}$ is a $\langle{\omega_n,N,\beta,\lambda_n+\delta'\rangle}$-strategy and we are done, so we assume that $N>\rho_n$. Note that if the algorithm ever uses $N\omega+\beta$ colors, then we can connect the components with those colors into one, and we are done. Therefore, we may assume that the algorithm never uses $N\omega+\beta$ colors. Let $m=\binom{N\omega-\rho_n\omega-1}{\omega_n}+3$. The strategy consists of two phases.
    
    In Phase I, we play the strategy $S_{\omega_n}$ repeatedly on disjoint sets of vertices until there are $m$ interval graphs $G_1,\ldots,G_m$ of clique number at most $\omega_n$ such that for some set $\mathcal{H}$ of $\rho_n\omega+\beta$ colors and for each $i\in\{1,\ldots,m\}$, the graph $G_i$ has a component $H_i$ that was assigned every color in $\mathcal{H}$, and an interval representation $I_i$ so that the image of each component of $G_i$ under $I_i$ is contained in an interval of length $\lambda_n+\delta'$. This is guaranteed by the pigeonhole principle since the algorithm cannot use $N\omega+\beta$ colors. 
    
    In Phase II, we construct $m-2$ cliques $K_{\omega_n}^3,\ldots,K_{\omega_n}^{m}$ of size $\omega_n$ so that if $i\in\{3,\ldots,m\}$ and $u\in K_{\omega_n}^i$, then $uv$ is an edge if and only if $v\in H_i$. By the pigeonhole principle, there are integers $a$ and $b$ such that $K_{\omega_n}^a$ and $K_{\omega_n}^b$ are assigned the same $\omega_n$ colors. Finally, we construct two cliques $K_{\omega_n}^{1}$ and $K_{\omega_n}^2$ of size $\omega_n$ so that if $u\in K_{\omega_n}^{1}$, then $uv$ is an edge if and only if $v\in H_{1}\cup K_{\omega_n}^a$, and if $u\in K_{\omega_n}^{2}$, then $uv$ is an edge if and only if $v\in H_{2}\cup K_{\omega_n}^{1}\cup K_{\omega_n}^b$. Figure \ref{fig:bridge} illustrates the component with a total of 
    \begin{align*}
    \rho_n\omega+\beta+3\omega_n
    &=[\rho+3(2^n-1)]\omega+\beta+3(2^n\omega)\\
    &=[\rho+3(2^n-1)+3(2^n)]\omega+\beta\\
    &=[\rho+3(2^{n+1}-1)]\omega+\beta\\
    &=\rho_{n+1}\omega+\beta
    \end{align*}
    colors under a $(4^n\lambda+\delta/4)$-interval representation $I$. Let $H$ denote this component. It is immediate from Figure \ref{fig:bridge} that the image of $H$ under $I$ is contained in an interval of length 
    \begin{align*}
    4(\lambda_n+d')+8d'&=4(4^n\lambda+\delta/12)+8\delta/12\\&=4^{n+1}\lambda+\delta
    \end{align*} 
    and has clique number $\omega_n+\omega_n=2^{n+1}\omega$. It is easy to verify that there is an interval representation $I$ so that the images of the remaining components under $I$ are contained in intervals of length $4^n\lambda+\delta$ and have clique number at most $2^{n+1}\omega$. 

\end{proof}

\begin{theorem}\label{thm1}
    Let $\rho$, $\beta$ and $\lambda$ be real numbers. For each positive integer $\omega$, let $S_\omega$ be a $\langle{\omega,\rho,\beta,\lambda\rangle}$-strategy. For every $\varepsilon>0$ and $\delta>0$, there is no $(3-\varepsilon)$-competitive online coloring algorithm for $\left({\lambda(3-\rho)^2}{\varepsilon^{-2}}+\delta\right)$-interval graphs.    
\end{theorem}

\begin{proof}
    Let $k$ be a positive integer, let $\rho$, $\beta$ and $\lambda$ be real numbers, and let $\varepsilon>0$ and $\delta>0$. If $\varepsilon\geq 3-\rho$, then $\rho\geq 3-\varepsilon$ and the proof is trivial. Therefore, we assume that $\varepsilon<3-\rho$ and choose $n=\lceil\log(\frac{3-\rho}{\varepsilon})\rceil$. Let $S$ be a $\langle{2^{k-n},\rho,\beta,\lambda\rangle}$-strategy. 
    By Lemma \ref{lemma:strat}, there is a $\langle{2^n2^{k-n}, \rho+3(2^n-1), \beta, 4^n\lambda+\delta\rangle}$-strategy that constructs a $(4^{n-1}\lambda+\delta)$-interval graph. Therefore, for every positive integer $k$, we force 
    \begin{align*}
    \bigl[\rho+3(2^n-1)\bigr]2^{k-n}+\beta
    &= \bigl[2^{-n}\rho+3(1-2^{-n})\bigr]2^k+\beta\\
    &= \bigl[2^{-n}\rho+3-3\cdot 2^{-n}\bigr]2^k+\beta\\
    &= \bigl[3-(3-\rho)\cdot 2^{-n}\bigr]2^k+\beta\\
    &= \bigl[3-(3-\rho)2^{-\lceil\log(\frac{3-\rho}{\varepsilon})\rceil}\bigr]2^k+\beta\\
    &\ge [3-\varepsilon]2^k+\beta
    \end{align*}
    colors on a $\left({\lambda(3-\rho)^2}{\varepsilon^{-2}}+\delta\right)$-interval graph of clique number at most $2^k$.
\end{proof}

Finally, in order to provide our best bounds shown in Figure \ref{fig:gap}, we present strategies inspired by those of Epstein and Levy \cite{eps-lev-05} and Micek \cite{mic-08,survey}. 

\begin{theorem}[Epstein and Levy \cite{eps-lev-05}]
    Let $\omega$ be a positive integer and let $\varepsilon>0$. Then, there is an $\langle\omega,1.5,-0.5,2+\varepsilon\rangle$-strategy. Moreover, the strategy constructs a unit-interval graph.
    \label{el}
\end{theorem}

\begin{proof}

    \begin{figure}
    \begin{tikzpicture}[scale=0.75]
        \draw[fill=red!20] (0,0) rectangle (4,2);
        \path[thick,pattern=horizontal lines, pattern color = white] (0,0) rectangle (4,2);
        \node at (2,1){$K_{\left\lfloor\omega/2\right\rfloor}$};
        \draw[fill=blue!20] (4.25,0) rectangle (8.25,2);
        \path[thick,pattern=horizontal lines, pattern color = white] (4.25,0) rectangle (8.25,2);
        \node at (6.25,1){$H$};
        \draw[fill=gray!10] (4.5,2.2) rectangle (8.5,4.2);
        \path[thick,pattern=horizontal lines, pattern color = white] (4.5,2.2) rectangle (8.5,4.2);
        \node at (6.5,3.2){$K_{\omega}-H$};
        \draw[fill=green!20] (0.25,2.2) rectangle (4.25,4.2);
        \path[thick,pattern=horizontal lines, pattern color = white] (0.25,2.2) rectangle (4.25,4.2);
        \node at (2.25,3.2){$K_{\left\lceil\omega/2\right\rceil}$};
        \draw[|-|] (0,-0.3) -- (8.5,-0.3);
        \draw[|-|] (-0.35,0) -- (-0.35,2);
        \draw[|-|] (-0.05,2.2) -- (-0.05,4.2);
        \node at (-1,1){$\left\lfloor\omega/2\right\rfloor$};
        \node at (-0.75,3.2){$\left\lceil\omega/2\right\rceil$};
        \node at (4.3,-1){$2+\varepsilon$};
    \end{tikzpicture}
    \caption{A unit-interval representation for the $\langle\omega,1.5,-0.5,2+\varepsilon\rangle$-strategy in Theorem \ref{el}.}
    \label{figeps}
    \end{figure}
    
    Let $\omega$ be a positive integer, let $\varepsilon'>0$ and let $\varepsilon=\min\{\varepsilon',0.1\}$. Present a clique $K_{\lfloor\omega/2\rfloor}$ of size $\lfloor\omega/2\rfloor$ and a clique $K_{\omega}$ of size $\omega$. The algorithm must assign $K_{\omega}$ at least $\lceil\omega/2\rceil$ colors not used on $K_{\lfloor\omega/2\rfloor}$. Let $H$ be a subset of $V(K_{\omega})$ of size $\lfloor\omega/2\rfloor$ with colors not used on $K_{\lfloor\omega/2\rfloor}$. Present one more clique $K_{\lceil\omega/2\rceil}$ of size $\lceil\omega/2\rceil$ with edges between each of its vertices and each of the vertices in $K_{\lfloor\omega/2\rfloor}$ and $H$, but no edges between its vertices and $K_{\omega}$. 
    
    We define a unit-interval representation $I$ for this graph by
    \[
    I(v)=
    \begin{cases}
    \left[0,1\right], & \text{if } v\in K_{\lfloor\omega/2\rfloor},\\[6pt]
    \left[1+\frac{\varepsilon}{2},2+\frac{\varepsilon}{2}\right], & \text{if } v\in H,\\[6pt]
    \left[1+\varepsilon,2+\varepsilon\right], & \text{if } v\in K_{\omega}-H,\\[6pt]
    \left[\frac{\varepsilon}{2},1+\frac{\varepsilon}{2}\right], & \text{if } v\in K_{\lceil\omega/2\rceil}.
    \end{cases}
    \]
    The interval representation $I$ is illustrated in Figure \ref{figeps}. Since $\varepsilon\leq \varepsilon'$, this concludes the proof.
    
\end{proof}

\begin{theorem}[Micek \cite{mic-08}, Bosek et al. \cite{survey}]
    Let $\omega$ be a positive integer and let $\varepsilon>0$. Then, there is a $\langle\omega,2,-1,3+\varepsilon\rangle$-strategy. Moreover, the strategy constructs a unit-interval graph.
    \label{mic}
\end{theorem}

\begin{proof}

    \begin{figure}
    \begin{tikzpicture}[x=3cm,y=0.45cm]
    
    \def\om{9}          
    \def\t{5}           
    \def\eps{0.4}       

    \draw[red, dashed, fill=red!10] (1.9,0.45) rectangle (3.07+\eps,\t+0.45);
    \draw[red, dashed, fill=red!10] (-0.13,0.45) rectangle (1+\eps,\t+0.45);
    \draw[blue, dashed, fill=blue!10] (0.22,0.2) rectangle (1.43,-\om+\t);
    \draw[amber, dashed, fill=amber!10] (2.78,0.2) rectangle (1.57,-\om+\t);
    \draw[green, dashed, fill=green!10] (0.9,\t+1.1) rectangle (2+\eps,2*\t+0.2);

    \node[green,right] at (2+\eps,3*\t/2+0.65){\(\scriptstyle t-1\)};
    \node[amber,right] at (2.78,-\om/2+\t/2+0.1){\(\scriptstyle \omega-t\)};
    \node[blue,left] at (0.22,-\om/2+\t/2+0.1){\(\scriptstyle \omega-t\)};
    \node[red,left] at (1.9,\t/2+0.45){\(\scriptstyle t\)};
    \node[red,left] at (-0.13,\t/2+0.45){\(\scriptstyle t\)};
    
    \foreach \i in {1,...,\t}{
        \pgfmathsetmacro{\L}{(\eps/\om)*\i}
        \pgfmathsetmacro{\R}{\L+1}
        \draw[thick] (\L,\i) -- (\R,\i);
        \draw (\L,\i-0.05) -- (\L,\i+0.05);
        \node[left] at (\L,\i) {\(\scriptstyle v_{\i}^0\)};
    }
    \pgfmathsetmacro{\tpone}{\t+1}
    \foreach \i in {\tpone,...,\om}{
        \pgfmathsetmacro{\y}{\om-\i+0.5}
        \pgfmathsetmacro{\L}{2/5}
        \pgfmathsetmacro{\R}{\L+1}
        \pgfmathtruncatemacro{\ii}{\i}
        \draw[thick] (\L,-\y) -- (\R,-\y);
        \node[left] at (\L,-\y) {\(\scriptstyle v_{\ii}^0\)};
    }

    \foreach \i in {1,...,\t}{
        \pgfmathsetmacro{\L}{2+(\eps/\om)*\i}
        \pgfmathsetmacro{\R}{\L+1}
        \draw[thick] (\L,\i) -- (\R,\i);
        \node[right] at (\R,\i) {\(\scriptstyle v_{\i}^2\)};
    }

    \foreach \i in {\tpone,...,\om}{
        \pgfmathsetmacro{\y}{\om-\i+0.5}
        \pgfmathsetmacro{\L}{8/5}
        \pgfmathsetmacro{\R}{\L+1}
        \pgfmathtruncatemacro{\ii}{\i}
        \draw[thick] (\L,-\y) -- (\R,-\y);
        \node[right] at (\R,-\y) {\(\scriptstyle v_{\ii}^2\)};
    }

    \foreach \i in {1,...,\the\numexpr\t-1\relax}{
        \pgfmathsetmacro{\L}{(\eps/(\om)*(\i+1/2) + 1}
        \pgfmathsetmacro{\R}{\L+1}
        \pgfmathsetmacro{\y}{\t+\i+0.5}
        \draw[thick] (\L,\y) -- (\R,\y);
        \draw[dotted] (\L,\y) -- (\L,-\om+\t-0.65);
        \draw[dotted] (\R,\y) -- (\R,-\om+\t-0.65);
        \node[left] at (\L,\y)
        {\(\scriptstyle v_{\i}^1\)};
    }
    
    \draw[dashed] (-0.25,-\om+\t-0.4) rectangle (1.47,\t+0.8);
    \draw[dashed] (3.14+\eps,-\om+\t-0.4) rectangle (1.53,\t+0.8);
    \node[right] at (-0.25,-\om+\t+0.1){$\scriptstyle K_\omega^0$};
    \node[left] at (3.14+\eps,-\om+\t+0.1){$\scriptstyle K_\omega^2$};
    
    \pgfmathsetmacro{\barl}{-\om+\t-0.95}
    \foreach \i in {0,1,2,3}{
        \draw (\i,\barl+0.1) -- (\i,\barl-0.1) node[below]{\(\scriptstyle\i\)};
    }
    \draw (3+\eps,\barl+0.1) -- (3+\eps,\barl-0.1) node[below]{\(\scriptstyle3+\varepsilon\)};
    \draw (-0.25,\barl) -- (3.14+\eps,\barl);
    \end{tikzpicture}
    \caption{A unit-interval representation for the $\langle\omega,2,-1,3+\varepsilon\rangle$-strategy in Theorem \ref{mic}.}
    \label{figmic}
    \end{figure}
    
    Let $\omega$ be a positive integer, let $\varepsilon'>0$ and let $\varepsilon=\min\{\varepsilon',0.1\}$. First, we present two cliques $K_\omega^0$ and $K_\omega^2$ of size $\omega$. Let $t$ denote the number of colors that $K_\omega^0$ and $K_\omega^2$ have in common. If $t\leq 1$, then we are done. Hence, we assume that $t\geq 2$. Let $V(K_\omega^j)=\{v^j_1,\ldots,v^j_t,v^j_{t+1},\ldots,v^j_\omega\}$ such that $v^0_i$ and $v^2_i$ were assigned the same colors for $i\in\{1,\ldots,t\}$. 
    
    We define a unit-interval representation $I$ for this graph by
    \[
    I(v_i^j)=
    \begin{cases}
    \left[j+\dfrac{\varepsilon}{\omega}i,\; j+\dfrac{\varepsilon}{\omega}i+1\right], & \text{if } i\le t,\\[6pt]
    \left[0.4,1.4\right], & \text{if } i> t\text{ and } j=0,\\[6pt]
    \left[1.6,2.6\right], & \text{if } i> t\text{ and } j=2.\\
    \end{cases}
    \]
    and force the remaining $t-1$ colors by presenting unit intervals 
    \[I(v_i^1)=\left[\frac{\varepsilon}{\omega}(i+0.5)+1,\;\frac{\varepsilon}{\omega}(i+0.5)+2\right]\text{ for }i\in\{1,\ldots,t-1\}.\] 
    The interval representation $I$ is illustrated in Figure \ref{figmic}. Since $\varepsilon\leq \varepsilon'$, this concludes the proof.
\end{proof}





\section{Online coloring of 2-count interval graphs}\label{sec:two}

In this section, we study the online coloring of 2-count interval graphs. First, we show that the competitive ratio of $\ff$ is bounded above by 4 in this setting. In general, we show that no online algorithm has competitive ratio less than 2.5 if the interval representation is unknown and 2 if the interval representation is known. In this setting, the interval length may depend on the clique number of the graph. Unlike for $\sigma$-interval graphs, 2-count interval graphs form a proper subclass of interval graphs without needing to bound the length.

\subsection{Upper bounds}

Given a graph $G$ and an ordering of its vertices $v_1,\ldots,v_n$, the algorithm $\ff$ assigns each vertex $v_i$ the smallest color in $\mathbb{N}$ not assigned to its neighbors in $\{v_1,\ldots,v_{i-1}\}$. The competitive ratio of $\ff$ on interval graphs is at most $8$ according to Naraswamy and Babu \cite{nar-sub-08}. We show that the upper bound can be improved to $4$ when restricted to 2-count interval graphs.

\begin{theorem}\label{thm:two-ff}
    $\ff$ is 4-competitive on 2-count interval graphs.
\end{theorem}

\begin{proof}
    We show that $\ff$ uses at most $4\omega -3$ colors on any sequence of intervals $x_1,\ldots,x_n$ where $\omega$ denotes the maximum clique size of the interval graph represented by $\{x_1,\ldots,x_n\}$. For each $i$, let $X_i=\{x_1,\ldots,x_i\}$  and let $\alpha_i\in\mathbb{N}$ denote the color assigned to $x_i$ by $\ff$. This means that $\alpha_i\leq|N(x_i)\cap X_{i-1}|+1$ where $N(x_i)$ denotes the neighborhood of $x_i$, equivalently the set of intervals intersecting $x_i$. Fix $i\leq n$. If $x_i$ is a shortest interval in $X_i$, then $x_i$ intersects at most $2\omega -2$ intervals in $X_{i-1}$. Therefore, $\alpha_i\leq 2\omega -1$. If $x_i$ is longer than some interval in $X_{i-1}$, then let $\gamma$ be the largest color assigned to an interval intersecting $x_i$. Let $j<i$ and let $x_j$ be any interval intersecting $x_i$ with $\alpha_j=\gamma$. If $x_j$ is shorter than $x_i$, then $\gamma\leq 2\omega -1$. If $x_j$ is the same length as $x_i$, then $x_j$ contains at least one endpoint of $x_i$. Therefore, there are at most $2\omega-2$ intervals of the same length as $x_i$ intersecting $x_i$ and $\alpha_i\leq 2\omega-1+|L|$ where $L$ is the set of intervals in $X_{i-1}$ with the same length as $x_i$ intersecting $x_i$. Thus, $\alpha_i\leq 4\omega-3$.
\end{proof}

On the other hand, the Kierstead-Trotter algorithm in \cite{kie-tro-81} still has competitive ratio $3$ in this setting.

\subsection{Unknown interval representation}
We show that for every positive integer $k$, there is a strategy that forces $5k-1$ colors on a 2-count interval graph of clique number at most $2k$.
\begin{theorem}
    There is no online algorithm that uses less than $\lfloor(5\omega-1)/2\rfloor$ colors on any 2-count interval graph of clique number at most $\omega$.
    \label{thm:two-len}
\end{theorem}

\begin{proof}
Let $k$ be a positive integer. We present a strategy that forces $5k-1$ colors on a 2-count interval graph of clique number at most $2k$. By Theorem \ref{mic}, there is a $\langle k,2,-1,3.1\rangle$-strategy $\B_k$ which constructs a unit-interval graph. Let $m=\binom{6k}{2k}+2$ and $n=m\binom{10k}{4k}+1$. The strategy consists of two phases.

In Phase I, we play the strategy $\B_k$ $n$ 
times resulting in $n$ disjoint subgraphs $M_1,\ldots,M_{n}$ so that for each $i\in\{1,\ldots,n\}$, 
$M_i$ has a unit-interval representation $I_i$ such that the image of $M_i$ under $I_i$ can be contained in the interval $[5i+20,5i+24]$ and $M_i$ is assigned $2k-1$ colors. 
By the pigeonhole principle, there are $m$ graphs $M_{i_1},\ldots,M_{i_{m}}$ which were assigned the same $2k-1$ colors by the algorithm.

In Phase II, we present $m-2$ cliques $K_k^3,K_k^4,\ldots,K_k^{m}$ of size $k$ so that if $v\in K_k^j$, then $vu$ is an edge if and only if $u\in M_{i_j}$. By the pigeonhole principle, there are two cliques $K_k^a$ and $K_k^b$ that were assigned the same $k$ colors by the algorithm. We force the last $2k$ colors by presenting two cliques $K_k^1$ and $K_k^2$ of size $k$ so that if $v\in K_k^1$, then $vu$ is an edge if and only if $u\in M_{i_1}\cup K_k^a$, and if $v\in K_k^2$, then $vu$ is an edge if and only if $u\in M_{i_2}\cup K_k^1\cup K_k^b$. In total, we force $2k-1+k+k+k=5k-1$ colors on a graph of clique number at most $2k$.

For $i\in\{i_a,i_1,i_2,i_b\}$, let $I'_i$ be a unit-interval representation of $M_i$ so that the image of $M_i$ under $I'_i$ is contained in the interval $(0,4),(4,8),(8,12)$ and $(12,16)$, respectively. We define a 2-count interval representation $I$ as follows.

\begin{figure}
    \centering
    \begin{tikzpicture}[yscale=0.5,xscale=0.45]

        \def\xA{-1}
        \def\xB{0}
        \def\xC{4}
        \def\xD{5}
        \def\xE{5}
        \def\xF{6}
        \def\xG{10}
        \def\xH{11}
        \def\xI{11}
        \def\xJ{12}
        \def\xK{16}
        \def\xL{17}
        \def\xM{17}
        \def\xN{18}
        \def\xO{22}
        \def\xP{23}
        
        \draw[dotted] (\xA,1.2)--(\xA,5.5) node[above]{$\scriptstyle 0$};
        \draw[dotted] (\xD,1.2)--(\xD,5.5) node[above]{$\scriptstyle 4$};
        \draw[dotted] (\xH,2.4)--(\xH,5.5) node[above]{$\scriptstyle 8$};
        \draw[dotted] (\xL,1.2)--(\xL,5.5) node[above]{$\scriptstyle 12$};
        \draw[dotted] (\xP,1.2)--(\xP,5.5) node[above]{$\scriptstyle 16$};
        
        \fill[yellow!50!] (\xB,0) rectangle (\xC,1.1);
        \node at ({(\xB+\xC)/2},0.5){$\scriptstyle M_{i_a}$};
        
        \fill[yellow!50!] (\xF,0) rectangle (\xG,1.1);
        \node at ({(\xF+\xG)/2},0.5){$\scriptstyle M_{i_1}$};
        
        \fill[yellow!50!] (\xJ,0) rectangle (\xK,1.1);
        \node at ({(\xJ+\xK)/2},0.5){$\scriptstyle M_{i_2}$};
        
        \fill[yellow!50!] (\xN,0) rectangle (\xO,1.1);
        \node at ({(\xN+\xO)/2},0.5){$\scriptstyle M_{i_b}$};
        
        \fill[red!40!] (\xA,1.2) rectangle (\xE,2.3);
        \path[thick,pattern=horizontal lines, pattern color = white] (\xA,1.2) rectangle (\xE,2.3);
        \node at ({(\xA+\xE)/2},1.7){$\scriptstyle K_{k}^{a}$};
        
        \fill[red!40!] (\xL,1.2) rectangle (\xP,2.3);
        \path[thick,pattern=horizontal lines, pattern color = white] (\xL,1.2) rectangle (\xP,2.3);
        \node at ({(\xL+\xP)/2},1.7){$\scriptstyle K_{k}^{b}$};
        
        \fill[blue!40!] (\xD,2.4) rectangle (\xI,3.5);
        \path[thick,pattern=horizontal lines, pattern color = white] (\xD,2.4) rectangle (\xI,3.5);
        \node at ({(\xD+\xI)/2},2.9){$\scriptstyle K_{k}^{1}$};
        
        \fill[green!40!] (\xH,3.6) rectangle (\xM,4.7);
        \path[thick,pattern=horizontal lines, pattern color = white] (\xH,3.6) rectangle (\xM,4.7);
        \node at ({(\xH+\xM)/2},4.1){$\scriptstyle K_{k}^{2}$};
        
        \draw[|-|] ({\xB-0.5},0) -- ({\xB-0.5},1.1);
        \node[left] at ({\xB-0.5},0.5){$\scriptstyle {k}$};
        
        \draw[|-|] ({\xA-0.5},1.2) -- ({\xA-0.5},2.3);
        \node[left] at ({\xA-0.5},1.7){$\scriptstyle {k}$};
        
        \draw[|-|] ({\xD-0.5},2.4) -- ({\xD-0.5},3.5);
        \node[left] at ({\xD-0.5},2.9){$\scriptstyle {k}$};
        
        \draw[|-|] ({\xH-0.5},3.6) -- ({\xH-0.5},4.7);
        \node[left] at ({\xH-0.5},4.1){$\scriptstyle {k}$};
        
        \foreach \L/\R in {\xB/\xC, \xF/\xG, \xJ/\xK, \xN/\xO}{
            \draw[|-|] (\L,-0.5) -- (\R,-0.5);
            \node[below] at ({(\L+\R)/2},-0.5){$\scriptstyle {3+\varepsilon}$};
        }
        
        
        \end{tikzpicture}
    \caption{A 2-count interval representation for the component assigned $5k-1$ colors in Theorem 3.2.}
    \label{fig:bridge2}
    \end{figure}

    \[
    I(v)=
    \begin{cases}
    I'_i(v), & \text{if } i\in\{i_a,i_1,i_2,i_b\} \text{ and }v\in M_i,\\[4pt]
    I_i(v), & \text{if } i\in\{1,\ldots,n\}-\{i_a,i_1,i_2,i_b\} \text{ and }v\in M_i,\\[4pt]
    \left[5i+20,\; 5i+24\right], & \text{if } i\in\{3,\ldots,m\}-\{a,b\} \text{ and }v\in K_k^i,\\[4pt]
    \left[0,4\right], & \text{if } v\in K_k^a,\\[4pt]
    \left[12,16\right], & \text{if } v\in K_k^b,\\[4pt]
    \left[4i,4i+4\right], & \text{if } i\in\{1,2\}\text{ and } v\in K_k^i.
    \end{cases}
    \]
Lastly, all intervals introduced in Phase I have length 1, and all intervals introduced in Phase II have length $4$. Figure \ref{fig:bridge2} illustrates the component that was assigned $5k-1$ colors.
\end{proof}

\subsection{Known interval representation}
In \cite{cgjmp-24}, Chybowska-Sok\'ol et al. present a strategy that forces $5k$ colors on a $(1+\varepsilon)$-interval graph with clique number $3k$ for any positive integer $k$ and $\varepsilon>0$. This strategy requires two interval lengths in $[1,1+\varepsilon]$ but can be carefully modified to use intervals of lengths in $\{1,1+\varepsilon\}$. They also present a strategy that forces $7k$ colors on a $(2+\varepsilon)$-interval graph with clique number $4k$ for any positive integer $k$ and $\varepsilon>0$. This strategy requires two interval lengths in $[1,2+\varepsilon]$. This strategy can be modified to use two interval lengths from an initially declared set of three lengths for any given $\omega>0$. 
Any known strategy that forces more than $\frac{7}{4}\omega+o(\omega)$ colors on an interval graph presented with interval representation requires at least three interval lengths. We present a strategy that forces at least $8k$ colors on a $2$-count interval graph of clique number $4k$.

First, we provide an outline of the proof. Our approach is similar to that of Lemma \ref{lemma:strat} and Theorem \ref{thm:two-len}. In Phase I, we aim to force four disjoint same-colored unit-interval graphs via a repeated sub-strategy and the pigeonhole principle. In Phase II, our goal is to introduce cliques of longer intervals and apply the ``bridging strategy'' illustrated in Figures 3 and 6. However, when the interval representation is known, the coloring algorithm may not provide the necessary positioning of the same-colored cliques. Hence, we first present two cliques covering the two outer same-colored unit-interval graphs from Phase I. We then consider two cases. If most of the colors in the two cliques are the same, then we continue with the same bridging strategy as before, introducing two more cliques. Otherwise, we instead introduce one clique ``bridging'' across the two middle same-colored unit-interval graphs from Phase I. Figure~7 illustrates the two cases within Phase II. It remains to address one final issue. If the graphs in Phase I are not positioned perfectly, we cannot guarantee the maximal number of colors with only two interval lengths. We solve this by spacing out the four graphs so that there are sufficiently large empty gaps between the middle graphs and the outer graphs from Phase I, and the two middle graphs from Phase I are very close to the midpoint of the two outer graphs. 

Let $\EL_\omega$ denote the strategy in \cite{eps-lev-05} that forces $\lfloor1.5\omega\rfloor$ colors on a sequence of intervals whose union can be contained in any interval of length $2+\varepsilon$ for any $\varepsilon>0$ where $\omega$ denotes the clique number of the graph represented. 

Let $s\in\mathbb{R}$, $\lambda>0$, and $\omega\in\mathbb{Z}_+$. We define the strategy $\mathrm{Gap}(\omega,s,\lambda)$ as follows. For each $i\in\mathbb{N}$, play $\EL_\omega$ to construct a unit-interval graph $E_i$ in the interval $(s+\lambda^i, s+\lambda^i+3)$ until no new colors have been used on $E_i$. We denote the set of colors assigned to $E_i$ by $\mathcal{E}_i$, and we refer to the index $t$ at which the strategy terminates as the \emph{terminal index} of the graph $E_1\cup\ldots\cup E_t$ constructed. This means that the strategy terminates at \emph{terminal index} $t$ if $\mathcal{E}_t\subseteq\mathcal{E}_1\cup\ldots\cup\mathcal{E}_{t-1}$. 

Notice that the first $t-1$ graphs are contained in the interval $(s,s+\lambda^{t-1}+3)$, and the last graph is contained in the interval $(s +\lambda^t,s+\lambda^t+3)$. This means that the interval $(s+\lambda^{t-1}+3,s +\lambda^t)$ does not intersect any interval introduced by $\mathrm{Gap}$. We can make this gap as big as necessary by choosing $\lambda$ sufficiently large. Additionally, if there is a bound on the number of colors that the coloring algorithm can use, then by the pigeonhole principle, $\mathrm{Gap}$ is guaranteed to terminate.

\begin{figure}
        \centering
        \begin{tikzpicture}[x=0.46cm,y=0.5cm]
            \draw[red!40, fill=red!40] (0,0) rectangle (0.6,0.5);
            \draw[red!40, fill=red!40] (12,0) rectangle (12.2,0.5);
            \draw[red!40, fill=red!40] (12.4,0) rectangle (12.6,0.5);
            \draw[red!40, fill=red!40] (24,0) rectangle (24.1,0.5);
            \draw[very thick] (-0.2,0) -- (24.2,0);
            \foreach \i in {0,10,14,24}{
                \draw[very thick] (\i,-0.75) -- (\i,0);
                }
            \foreach \i in {0,2,...,24}{
                \draw[thick] (\i,-0.25) -- (\i,0);
                }
            \node[red] (ap) at (2.6,1.5){$E_{p,1},\ldots,E_{p,t_p-1}$};
            \draw[red,dotted] (0,1.2) -- (0,0.55);
            \draw[red,dotted] (4,1.2) -- (0.6,0.55);
            \draw[red,dotted] (3,1.2) -- (0.5,0.55);
            \draw[red,dotted] (2,1.2) -- (0.4,0.55);
            \draw[red,dotted] (1.5,1.2) -- (0.3,0.55);
            \draw[red,dotted] (1,1.2) -- (0.2,0.55);
            \draw[red,dotted] (0.5,1.2) -- (0.1,0.55);
            \node[red] (bp) at (24.6,1.5){$E_{p,t_p}$};
            \draw[red,dotted] (24.5,1.2) -- (24,0.55);
            \node[red] (aq) at (9.4,1.5){$E_{q,1},\ldots,E_{q,t_q-1}$};
            \draw[red,dotted] (11.2,1.2) -- (12.2,0.55);
            \draw[red,dotted] (10.2,1.2) -- (12.2,0.55);
            \draw[red,dotted] (9.2,1.2) -- (12.2,0.55);
            \draw[red,dotted] (8.2,1.2) -- (12.2,0.55);
            \draw[red,dotted] (7.2,1.2) -- (12.2,0.55);
            \node[red] (bq) at (13.3,1.5){$E_{q,t_q}$};
            \draw[red,dotted] (13.1,1.2) -- (12.4,0.55);
            \node (lp) at (0,-1.5){$\scriptstyle s_p$};
            \node (rp) at (24,-1.5){$\scriptstyle s_p+(\Lambda_{p})^{t_p}$};
            \node (lpd) at (10,-1.5){$\scriptstyle s_p+\frac{5}{12}(\Lambda_{p})^{t_p}$};
            \node (rpd) at (14,-1.5){$\scriptstyle s_p+\frac{7}{12}(\Lambda_{p})^{t_p}$};
            \node (mp) at (12,-3){$s_{p+1}$};
            \draw[very thick] (0,2.5) rectangle (10,3.5);
            \node (A) at (5,3){$\mathcal{A}$};
            \draw[very thick] (14.1,2.5) rectangle (24.1,3.5);
            \node (B) at (19.05,3){$\mathcal{L}$};
            \draw[dashed, blue] (-0.5,4) rectangle (24.5,10);
            \node (c1) at (2,5){If $|\mathcal{A}\cap \mathcal{L}|<k$};
            \draw[very thick] (8,4.5) rectangle (18,5.5);
            \node (C) at (13,5){$\mathcal{N}$};
            \draw[dashed, blue] (-0.5,6) -- (24.5,6);
            \node (c1) at (2,9){If $|\mathcal{A}\cap \mathcal{L}|\geq k$};
            \draw[very thick] (2.3,6.5) rectangle (12.3,7.5);
            \node (C2) at (7.3,7){$\mathcal{I}$};
            \draw[very thick] (12.3,7.5) rectangle (22.3,8.5);
            \node (D) at (17.3,8){$\mathcal{F}$};
            \draw[dashed] (12.3,0) -- (12.3,10.7);
            \node (mq) at (12.3,11){$s_{q+1}$};
            \draw[dashed] (12,0) -- (12,-2.7);
            \draw[blue] (-0.7,2.3) rectangle (24.7,10.2);
            \draw[red] (-0.7,-2.2) rectangle (25.5,2.2);
            \node[rotate=90,blue] at (-1.2,6.25) {Phase II};
            \node[rotate=90,red] at (-1.2,0) {Phase I};
            \draw[->, blue] (25.2,2.3) -- (25.2,10.2);
            
        \end{tikzpicture}
        \caption{Final result of the strategy in Theorem \ref{thm:final}}
        \label{fig:final}
    \end{figure}

\begin{theorem}
    There is no online algorithm that uses less than $2\omega-1$ colors on any 2-count interval graph of clique number at most $\omega$ with known interval representation. 
    \label{thm:final}
\end{theorem}

\begin{proof}

    Let $k$ be a fixed positive integer. We present a strategy that forces $8k$ colors on a 2-count interval graph of clique number $4k$ constructed via its interval representation. If the algorithm ever uses $8k$ colors, then we are done, so by the pigeonhole principle, we may assume that the strategy $\mathrm{Gap}$ always terminates. Choose $n$ and $\lambda$ to be sufficiently large [$n=\binom{8k}{3k}+1$ and $\lambda=67$ suffice]. For $i\in\{1,\ldots,n\}$, define $\Lambda_i$ by 
    \(
    \Lambda_i:=1+\lambda+\lambda^2+\ldots+\lambda^{n-i}.
    \)
    The strategy consists of two phases, and Phase I consists of at most $n$ subphases. 
    
    In Phase I, we initialize $s_1=0$, and in Subphase $i$, we play $\mathrm{Gap}(2k,s_i,\Lambda_{i})$. Let $t_i$ be the terminal index of graph 
    \(
    E_{i,1}\cup\ldots\cup E_{i,t_i}
    \)
    constructed in Subphase $i$, let $s_{i+1}=s_i+\frac{1}{2}\Lambda_{t_i}$, and for $j\in\{1,\ldots,t_i\}$, let $\mathcal{E}_{i,j}$ denote the set of colors assigned to $E_{i,j}$. If $\mathcal{E}_{i,t_i}=\mathcal{E}_{j,t_j}$ for some $j<i$, then we end Phase I and continue onto Phase II. Otherwise, we move onto Subphase $i+1$.
    By the pigeonhole principle, we may assume that Phase I ends with terminal indices $t_1,\ldots,t_q$ and $\mathcal{E}_{t_q}=\mathcal{E}_{t_p}$ for some $p$ and some $q$ with $p<q\le n$. 

    In Phase II, we introduce $2k$ copies of the intervals
    \[
    \left[s_p,\;s_p+\frac{5}{12}(\Lambda_{p})^{t_p}\right]
    \text{ and }
    \left[s_p+\frac{7}{12}(\Lambda_{p})^{t_p}+3,\;s_p+(\Lambda_{p})^{t_p}+3\right].
    \]
    Let $\mathcal{A}$ and $\mathcal{L}$ denote the set of colors assigned to the first set of $2k$ intervals and the set of colors assigned to the second set of $2k$ intervals.

    If $|\mathcal{A}\cap \mathcal{L}|<k$, we introduce $2k$ copies of the interval 
    \[
    \left[s_p+\frac{3}{12}(\Lambda_{p})^{t_p},\;s_p+\frac{8}{12}(\Lambda_{p})^{t_p}\right].
    \]
    Let $\mathcal{N}$ denote the set of colors assigned to the set of $2k$ intervals. Then, we force at least 
    \begin{align*}
    |\mathcal{E}_{q,t_q}|+|\mathcal{A}\cup \mathcal{L}|+|\mathcal{N}| &= 3k+|\mathcal{A}|+ |\mathcal{L}|-|\mathcal{A}\cap\mathcal{L}|+|\mathcal{N}| \\
    &\ge 3k+2k+2k-k+2k \\
    &= 8k
    \end{align*}
    colors. 
    
    If $|\mathcal{A}\cap \mathcal{L}|\geq k$, we introduce $2k$ copies of the intervals 
    \[
    \left[s_{q+1}-\frac{5}{12}(\Lambda_{p})^{t_p},\;s_{q+1}\right]
    \text{ and }
    \left[s_{q+1},\;s_{q+1}+\frac{5}{12}(\Lambda_{p})^{t_p}\right].
    \]
    Let $\mathcal{I}$ and $\mathcal{F}$ denote the sets of colors assigned to the first set of $2k$ intervals and second set of $2k$ intervals, respectively. Since $|\mathcal{I}\cap \mathcal{A}|=0$, $|\mathcal{I}\cap \mathcal{L}|< k$. Therefore, we force at least 
    \begin{align*}
    |\mathcal{E}_{q,t_q}|+|\mathcal{I}\cup \mathcal{L}|+|\mathcal{F}| &= 3k+|\mathcal{I}|+ |\mathcal{L}|-|\mathcal{I}\cap\mathcal{L}|+|\mathcal{F}| \\ &\ge 3k+2k+2k-k+2k \\ &= 8k
    \end{align*}
    colors.

    Lastly, all intervals introduced in Phase I have length 1, and all intervals introduced in Phase II have length $\frac{5}{12}(\Lambda_{p})^{t_p}$.
    
\end{proof}

\section*{Acknowledgments}
\thanks{This work was partially supported by the International Collaborative Research Program of Wenzhou-Kean University [ICRPSP2025002].}

\bibliographystyle{acm}
\bibliography{semi}

\end{document}